%% file: fd_gatdr_workshop.tex
\documentclass[acmsmall,screen,authorversion,nonacm]{acmart}

\usepackage[utf8]{inputenc} 
\usepackage[T1]{fontenc}    
\usepackage{hyperref}       
\usepackage{url}            
\usepackage{booktabs}       
\usepackage{amsfonts}       
\usepackage{nicefrac}       
\usepackage{microtype}      
\usepackage{xcolor}         
\usepackage{algorithm}
\usepackage{algorithmic}
\usepackage{float}
\usepackage{hhline}
\usepackage{amsmath,bm}
\usepackage{amsthm}
\usepackage{graphicx}
\usepackage[shortlabels]{enumitem}
\usepackage{tabularx}
\usepackage[english]{babel}

\newtheorem{definition}{Definition}

\usepackage{pifont}
\input{symbols_commands}
\allowdisplaybreaks[4]
\AtBeginDocument{%
  \providecommand\BibTeX{{%
    \normalfont B\kern-0.5em{\scshape i\kern-0.25em b}\kern-0.8em\TeX}}}

\begin{document}

\title{FD-GATDR: A Federated-Decentralized-Learning Graph Attention Network for Doctor Recommendation Using EHR}




\author{Luning Bi}
\email{luningbi@iastate.edu}
\affiliation{%
  \institution{Iowa State University}
  \city{Ames}
  \country{USA}}

\author{Yunlong Wang}
\email{yunlong.wang@iqvia.com}
\affiliation{%
  \institution{IQVIA}
  \city{Plymouth Meeting, PA}
  \country{USA}
}

\author{Fan Zhang}
\email{fan.zhang@iqvia.com}
\affiliation{%
  \institution{IQVIA}
  \city{Plymouth Meeting, PA}
  \country{USA}
}

\author{Zhuqing Liu}
\email{liu.9384@osu.edu}
\affiliation{%
  \institution{The Ohio State University}
  \city{Columbus, OH}
  \country{USA}
}

\author{Yong Cai}
\email{yong.cai@iqvia.com}
\affiliation{%
  \institution{IQVIA}
  \city{Plymouth Meeting, PA}
  \country{USA}
}

\author{Emily Zhao}
\email{emily.zhao@iqvia.com}
\affiliation{%
  \institution{IQVIA}
  \city{Plymouth Meeting, PA}
  \country{USA}
}

\renewcommand{\shortauthors}{Bi, et al.}

\begin{abstract}
In the past decade, with the development of big data technology, an increasing amount of patient information has been stored as electronic health records (EHRs). Leveraging these data, various doctor recommendation systems have been proposed. Typically, such studies process the EHR data in a flat-structured manner, where each encounter was treated as an unordered set of features. Nevertheless, the heterogeneous structured information such as service sequence stored in claims shall not be ignored. This paper presents a doctor recommendation system with time embedding to reconstruct the potential connections between patients and doctors using heterogeneous graph attention network. Besides, to address the privacy issue of patient data sharing crossing hospitals, a federated decentralized learning method based on a minimization optimization model is also proposed. The graph-based recommendation system has been validated on a EHR dataset. Compared to baseline models, the proposed method improves the AUC by up to 6.2\%. And our proposed federated-based algorithm not only yields the fictitious fusion center's performance but also enjoys a convergence rate of O(1/T).
\end{abstract}



\maketitle

\section{Introduction} \label{sec: intro}

Electronic health record (EHR) system has been growing rapidly in the past decade. EHR contains patients’ medical information, history, diagnoses, medications, treatment plans,and laboratory and test results. It provides automate and streamline workflow to facilitate the decision-making of providers. Benefiting from the progress in the machine learning area, a variety of deep learning techniques and frameworks have been applied to clinical applications including such as information extraction, representation learning, outcome prediction, phenotyping, and a de-identification \cite{shickel2017deep}. This study aims to build a doctor recommendation system using patients' EHR  history.

There are different types of EHR code representation applied in the existing studies. Word2vec methods have gained more popularity recently \cite{bai2017joint} \cite{che2017exploiting}. In the word-level vector representation methods, all possible meanings of a word are transformed as a single vector representation, ignoring the fact that the word represent differently in different context. To solve this problem, Alsentzer et al. proposed the clinic BERT, a context-based model that was trained on all notes from MIMIC III, which have achieved the state-of-art performance \cite{alsentzer2019publicly}. By using these embedding methods, patients can be represented by the sequence of clinic service codes. Techniques are ranged from autoencoders for diabetic nephropathy \cite{katsuki2018risk}, CNNs for unplanned readmission prediction \cite{pham2016deepcare}, LSTMs for heart failure prediction \cite{maragatham2019lstm}, and GRU networks for predicting mortality \cite{sankaranarayanan2021covid}\cite{shickel2017deep}. However, in these studies, EHR data were treated as flat-structured information. The relationship among services, patients and doctors were ignored. And in terms of prediction tasks, there is lack of studies explored the application of EHR data mining to doctor recommendation. Compared to treatment recommendation, doctor recommendation needs more structured information due to some factors such as speciality and the patient's preference. Therefore, a method that can learn underlying information from the hidden structure in EHR data is in in immediately needs for the doctor recommendation.

Standard machine learning approaches require centralizing the training data on one machine or in a data center. However, collecting clinical datasets from isolated medical centers is unpractical since Clinic data such as disease symptoms and medical recordings are highly sensitive \cite{lu2019learn}. Under this context, federated Learning (FL), which enables local client model to collaboratively learn a shared prediction model while keeping all the training data private, eliminates the need to store the data in the cloud. However, most of the existing works in FL are limited to systems with i.i.d. datasets and centralized parameter servers. In the real world, the EHR data is fully decentralized among hospitals. The data is non i.i.d. due to the geographic influence on human health. Nevertheless, each hospital can share de-identified, non-sensitive, and intermediate statistics with its neighborhood hospitals. In this study, a federated decentralized learning is proposed to improve the client models' performance while satisfying the privacy requirements. 

Our contributions can be summarized as follows:
\begin{itemize}
    \item To extract the structured information from the EHR data, a heterogeneous graph consisting different types of nodes and edges is built.
    \item We propose a heterogeneous graph attention network (HGAT) considering time sensitivity and node heterogeneity for the representation of the patient, doctor and service .
    \item A federated decentralized learning algorithm is proposed to address the issue of data sharing among hospitals.
    \item The case study shows that the proposed graph model can achieve better performance than other baseline models. The federated decentralized learning algorithm can realize comparable performance compared to the global training.
\end{itemize}

\section{Doctor Recommendation Using EHR}
In EHR, each interaction between patients and doctors is trying to answer two questions: "what is happening" and "what happens next" \cite{pham2017predicting}. The first question is about diagnosis of patients' health status. It also provides the answer to the second question by indicating the future disease risk and corresponding treatments. Traditionally, this step is completed by experienced doctors, which means expensive communication cost and time cost. Under this setting, EHR, which contains detailed information of patient medical history, has become an alternative. Therefore, this paper focuses on developing an end-to-end prototype for recommending a doctor based on EHR of the patients. However, there are two major challenges, i.e., data heterogeneity and data privacy.

\subsection{Data Heterogeneity}\label{sec: hg}

Although in some way EHR is similar to human natural language, EHR contains a variety of information such as billing codes, medical service codes, laboratory measurements, patient demographic information, doctor's speciality, etc. Utilizing part of the EHR as flat-structured data for the representation learning can lose a lot of valuable information. This study is aimed to extract heterogeneous structured from the EHR data. The concept is introduced as follows.

\begin{definition}
  \textbf{Heterogeneous Graph}. A heterogeneous graph can be represented by $G=(V,E,A,R,\phi,\varphi)$, where $V$ denotes the set of nodes and $E$ represents the set of edges, $A$ and $R$ represent the node and edge types respectively, and $\phi$ and $\varphi$ denote the the mapping functions $\phi:V \overrightarrow{} A$ and $\varphi:E \overrightarrow{} R$. Moreover, in HIN, $|A|+|R|\geq 3$.
\end{definition}

For better understanding, Fig. \ref{fig:graph_example} is illustrated as an example of EHR graph. There are three types of nodes: patient, doctor and service. The patient node have the attributes such as age, sexuality and location. The doctor node have the attributes such as speciality and location. The service nodes can be divided into three sub-level node types: diagnosis, procedure and product. Each type of service node are stored as text in EHR, for example, "px-RADIOGRAPHY: SHOULDER AND UPPER ARM". The text or service codes can be used for diagnosis, risk prediction and decision-making of the following services/treatments. The edges of the graph represents the interactions between nodes. 

\begin{figure}[h]%
\centering
\includegraphics[width=0.42\textwidth]{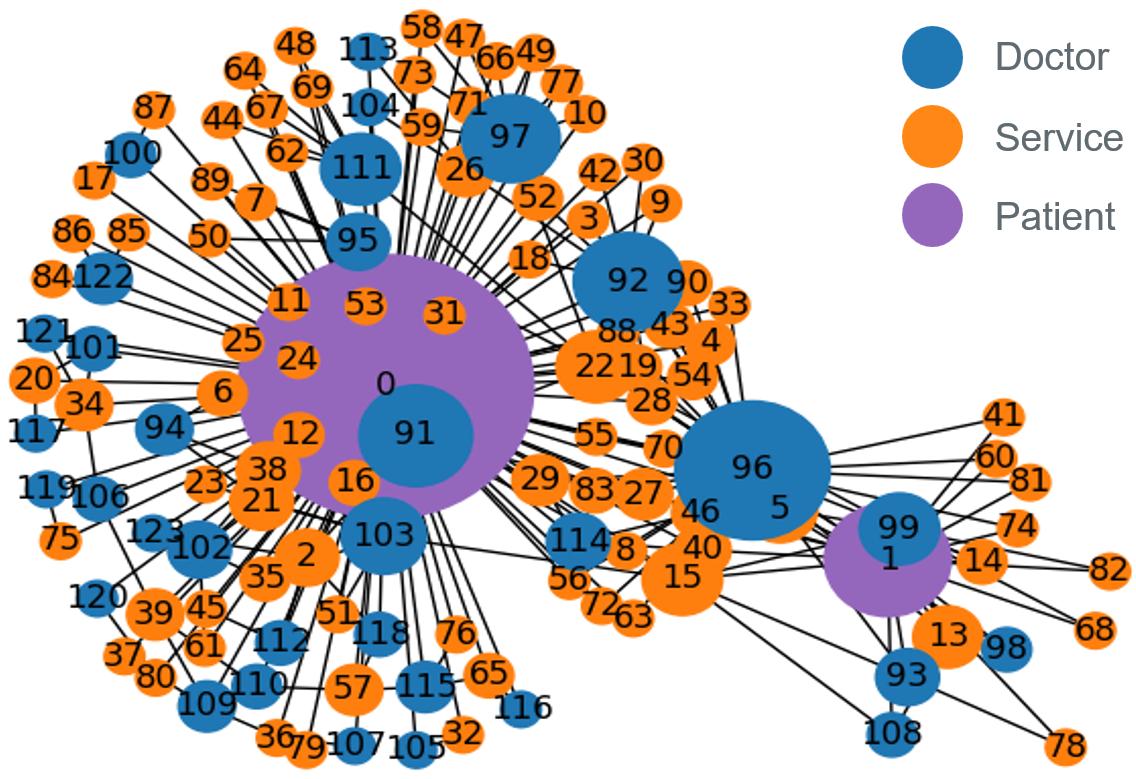}
\caption{An example of EHR graph. The service nodes can be divided into diagnosis, procedure and product.}\label{fig:graph_example}
\end{figure}

Four edge types are considered in this paper: patient-service, doctor-service, service-service and patient-doctor. 
\begin{itemize}

\item Patient-services means that the patient receives the service. 

\item Doctor-service characterizes the speciality of the doctor. 

\item Service-service represents the chronological order of services. For example, some products are usually prescribed together for a specific disease. And a diagnosis service will be followed the corresponding products and procedure, or even another diagnosis requirement. These structured information provides the basic logic of the doctor recommendation because it not only connects the patient and doctor but also connect the present with the future by predicting the next services. 

\item Patient-doctor is the target of doctor recommendation. The prediction of patient-doctor edge is based on the learned pattern from the aforementioned node attributes and the other three types of edge. In other words, doctor recommendation is based on the current status and future prediction. Based on the heterogeneous graph, the doctor recommendation considered system can be defined as follows.
\end{itemize}

\begin{definition}
  \textbf{Doctor Recommendation (DR)}. Given a EHR set of $<\mathcal{V_{P},V_{D},V_{S},W,B}>$ where $\mathcal{V_{D}}$ denotes the set of doctors, $\mathcal{V_{S}}$ denotes the set of services, $\mathcal{V_{P}}$ denotes the set of patients, $\mathcal{W}$ represents the attributes of the objects, and $\mathcal{B}$ represents the interactions between two types of objects. The DR system is aimed to predict the next couple of services $<\mathcal{V_{S}}>$ needed in future as well as the suitable doctors $<\mathcal{V_{P}}>$. For the implementation, a simple scheme is to mask out a most recent part, e.g., 20\%, of the EHR data and use the doctors appeared in the masked data as the ground-truth label.
\end{definition}

\subsection{Data Privacy} \label{sec: privacy}
The performance of GL models are largely determined by the amount of the training data. However, there are several challenges to this exercise. First, data privacy and security are sensitive in the medical area. It is difficult to integrate the data collected and aggregated across hospitals. Second, the EHR data collected across hospitals are often non i.i.d. due to the environmental factors. That means using only one hospital data may influence the generalization ability of the model. This paper will address these two concerns using a novel training method.

\section{Federated-Decentralized-Learning Graph Attention Network}
To solve the problems aforementioned, FD-GATDR, a federated-decentralized-learning graph attention network for doctor recommendation is proposed. 

\subsection{Service Embedding}
The first step is to convert the text of services into the representation in the computer language. Current research borrows the concepts from NLP, in which service codes are viewed as “words” and patients or encounters as the “sentences”. Sequences of patient encounters can also be seen as the “documents”. 
Deriving vector-based representation of clinical concepts is a common way to reduce the dimension of the code space and reveal complex relationships between different types of nodes, such as word2vec \cite{guthrie2006closer}, GLoVe \cite{pennington2014glove} and BERT \cite{devlin2018bert}. In the case study, all the three methods have been tested. It should be noted that the service embedding is just the first step to represent service. The service representation will be further updated by the graph learning model which is introduced in the following section. 

\subsection{Heterogeneous Graph Attention Network (HGAT)}
The basic idea of the proposed model HGAT is to learn meaningful and robust representations for different types of nodes. The input to the model is a set of node features, $\mathbf{h}=\{\vec{h}_{1},\vec{h}_{2},...,\vec{h}_{N}\}$, $\vec{h}_{i} \in \mathbb{R}^{F_{t_i}}
$, where $N$ denotes the number of nodes, $t_i$ represents the node type of the $\vec{h}_{i}$ and $F_{t_i}$ is the number of features for the node type of the $\vec{h}_{i}$.

\textbf{Neighborhood nodes sampling}. Different from other recommendation system, doctor recommendation is more time-sensitive, meaning that it should recommend the most needed services or doctors. It is assumed that the closer interactions have more prevalence with recommendation. Thus, instead of using first-order neighbor nodes or selecting nodes randomly, the neighbor nodes is selected based on the weight of edges using the roulette method. Specifically, the weights can be calculated based on the timestamp of the interaction using a custom-built linear or log method.

\textbf{Heterogeneous node representation}. To address the data heterogeneity issue, Eq. \ref{Eq: hattention} is proposed to obtain the attention coefficients. 

\begin{equation} \label{Eq: hattention}
\alpha_{i j}=\frac{\exp \left(\operatorname{LeakyReLU}\left(\overrightarrow{\mathbf{a}}^{T}\left[\mathbf{Q} \vec{h}_{i} \| \mathbf{Q} \vec{h}_{j}\| \mathbf{V}_{t_{i}t_{j}}\right]\right)\right)}{\sum_{k \in \mathcal{N}_{i}} \exp \left(\operatorname{LeakyReLU}\left(\overrightarrow{\mathbf{a}}^{T}\left[\mathbf{Q} \vec{h}_{i} \| \mathbf{Q} \vec{h}_{k}\| \mathbf{V}_{t_{i}t_{j}}\right]\right)\right)}
\end{equation}
Where $\mathbf{Q}$ denotes input linear transformation’s weight matrix and $\mathbf{V}_{t_{i}t_{j}}$ is a shared one-dimension vector that represents the attention between two nodes in the node type level. Compared the original version in \cite{velickovic2017graph}, the attention between node types are specified to improve the model performance on heterogeneous graph.

Specifically, we repeat the attention for $K$ times and concatenate the learned embeddings as the embedding:
\begin{equation}
\vec{h}_{i}^{\prime}=\sigma\left(\frac{1}{K} \sum_{k=1}^{K} \sum_{j \in \mathcal{N}_{i}} \alpha_{i j}^{k} \mathbf{Q}^{k} \vec{h}_{j}\right)
\end{equation}

Where $k$ represent concatenation, $\alpha_{i j}^{k}$ are normalized attention coefficients computed by the $k$-th attention mechanism, and $\mathbf{Q}^{k}$ is the corresponding input linear transformation’s weight matrix.

\subsection{Reconstruction Loss}

In this paper, it is assumed that there is an edge between every two successive services in the EHR of a patient. The robustness of an service-service edge is determined by the frequency of the edge in different patients' EHR. Moreover, instead of being an attribute of the doctor node, the doctor's speciality is masked out and used as a prediction task by using the doctor embedding obtained by HGAT. 

Therefore, three types of reconstruction loss are considered: doctor-specialty, service-service and patient-doctor as in Eq. \ref{Eq: overall}. The weights of each loss term can be manually set or using a Bayesian task weight learner introduced in \cite{zheng2021multi}.

\begin{equation}\label{Eq: overall}
\mathcal{L}=\mathcal{L}_{ds}+\mathcal{L}_{ss}+\mathcal{L}_{pd}
\end{equation}
Where $\mathcal{L}_{ds}$ denotes the reconstruction loss of doctor-speciality, $\mathcal{L}_{ss}$ denotes the reconstruction loss of service-service and  $\mathcal{L}_{pd}$ denotes the reconstruction loss of patient doctor. 

$\mathcal{L}_{ds}$ and  $\mathcal{L}_{ss}$ can be calculated the cross entropy loss function since they are multi-class prediction tasks.  For the patient-doctor reconstruction, a pairwise personalized ranking loss, namely, Bayesian personalized ranking (BPR) \cite{rendle2012bpr}, is used. The BPR loss is maximizing the match degree between the patient and the doctor recommended (positive sample) while minimizing the match degree between the patient and the doctor not recommended (negative sample). The model performance is evaluated using recall and AUC.

\subsection{Federated Decentralized Learning} \label{sec: alg}

To address the data privacy issue discussed in \ref{sec: privacy}, a federated decentralized learning method is proposed. In the federated decentralized scenario, the hospitals form a peer-to-peer network system, which can be represented by an undirected connected graph $\mathcal{G} = (\mathcal{N}, \mathcal{L})$.
Here, $\mathcal{N}$ and $\mathcal{L}$ are the sets of workers for local hospitals and the edges between hospitals, respectively, with $|\mathcal{N}| = m$. 
The workers are capable of local computation and communicating with their neighboring workers via the edges in $\mathcal{L}$.
The goal of fully decentralized FL is to have the workers {\em distributively} and {\em collaboratively} solving the global optimization problem in the following form:

\begin{align}\label{Eq: problem}
 \min_{\x \in \mathbb{R}^p} \frac{1}{m}\sum_{i=1}^{m} f(\x_i)= \min_{\x \in \mathbb{R}^p} \frac{1}{mn}\sum_{i=1}^{m}\sum_{j=1}^{n} f(\x_i;\zeta_{j}),
\end{align}

where each local objective function $f_i(\x) \triangleq \frac{1}{n}\sum_{j=1}^{n} f_i(\x;\zeta_{j})$ is only observable to worker $i$ and not necessarily convex.
Here, $\D_i$ represents the distribution of the dataset at node $i$, which is {\em heterogeneous} across workers.

First, we reformulate Problem~(\ref{Eq: problem}) in the following equivalent form by introducing a local model copy at each worker:
\begin{align}\label{Eq: consensus_problem}
	& \text{Minimize} && \hspace{-.5in} \frac{1}{m}\sum_{i=1}^{m} f(\x_i) & \\
	& \text{subject to} && \hspace{-.5in} \x_i = \x_{i'}, && \hspace{-.5in} \forall (i,i') \in \mathcal{L}. \nonumber
	\vspace{-.05in}
\end{align}

where $\x \triangleq [\x_1^\top,\cdots,\x_m^\top]^\top,$ and $\x_i$ is an introduced local copy at worker $i$.
To solve Problem \eqref{Eq: consensus_problem}, we consider an $\epsilon^2$-stationary point $\x$ defined as follows:
\begin{align}\label{Eq: FOSP_network}
	\underbrace{\Big\|\frac{1}{m}\sum_{i=1}^{m} \nabla f(\xb) \Big\|^2}_{\mathrm{Global \,\, gradient \,\, magnitude}} \!\!\!\! + \underbrace{\frac{1}{m}\sum_{i=1}^{m}\|\x_{i}- \xb\|^2}_{\mathrm{Consensus \,\, error}} \le \epsilon^2,
\end{align}

where $\xb \triangleq \frac{1}{m}\sum_{i=1}^{m} \x_{i}$ represents the global average across all workers.
The first term is the gradient norm of the global loss function and the second term is the average consensus error across all local copies.
In this work, we aim to develop an efficient algorithm to attain an $\epsilon^2$-stationary point for fully decentralized FL.

To solve Problem~\eqref{Eq: problem} in decentralized network systems where workers reach a {\em consensus} on a global optimal solution, a common approach in the literature is to let workers aggregate neighboring information through a consensus matrix $\W \in \mathbb{R}^{m\times m}$.
Let $[\W]_{ij}$ represent the element in the $i$-th row and the $j$-th column in $\W$.
Then, a consensus matrix $\W$ should satisfy the following properties:
\begin{enumerate}[topsep=1pt, itemsep=-.1ex, leftmargin=.35in]
	\item[(a)] {\em Doubly Stochastic:} $\sum_{i=1}^{m} [\mathbf{W}]_{ij}=\sum_{j=1}^{m} [\mathbf{W}]_{ij}=1$.
	\item[(b)] {\em Symmetric:} $[\mathbf{W}]_{ij} = [\W]_{ji}$, $\forall i,j \in \mathcal{N}$. 
	\item[(c)] {\em Network-Defined Sparsity Pattern:} $[\W]_{ij} > 0$ if $(i,j)\in \mathcal{L};$ otherwise $[\mathbf{W}]_{ij}=0$, $\forall i,j \in \mathcal{N}$.
\end{enumerate}
The above properties imply that the eigenvalues of $\W$ are real and can be sorted as $-1 < \lambda_m(\W) \leq \cdots \leq \lambda_2(\W) < \lambda_1(\W) = 1$.
We define the second-largest eigenvalue in magnitude of $\W$ as $\lambda \triangleq \max\{|\lambda_2(\W)|,|\lambda_m(\W)|\}$ for further notation convenience.
It can be seen later that $\lambda$ plays an important role in the step-size selection and characterizing the algorithm's convergence rate.

We start with stating the following assumptions:
\begin{assum}\label{Assumption: function}
The objectives $f(\cdot)$ and $f_i(\cdot)$ satisfy:
\begin{enumerate}[topsep=1pt, itemsep=-.1ex, leftmargin=.2in]
	\item[(1)] $f(\x)$ is bounded from below, i.e., there exists an $\x^*\in\mathbb{R}^p,$ such that $f(\x)\ge f(\x^*)$, $\forall \x\in\mathbb{R}^p;$
	\item[(2)] The function $f_i(\x)$ is continuously differentiable and has $L$-Lipschitz continuous gradients, i.e., there exists a constant $L >0$ such that $|\nabla f_i(\x_1) -\nabla f_i(\x_2)|\le L \| \x_1-\x_2 \|_2,$ $\forall \x_1,\x_2\in\mathbb{R}^p;$
\end{enumerate}
\end{assum}

\begin{algorithm}[t!]
	\caption{Federated Decentralized Learning}\label{Algorithm}
	\begin{algorithmic}[1]
		\STATE Set $\x_{i,0} = \x^0$ and $\y_{i,0} = \g_{i,0} =  \nabla f_i(\x_{i,0})$ at worker $i$, for all $i\in[m]$.
		\FOR{$k = 0, \cdots, K-1$}
		\FOR{worker $i$, $i \in [m]$}
		\STATE Share $(\x_{i,k}, \y_{i,k})$ with neighboring nodes;
		\STATE Consensus Update $\x_{i,k}  = {{\sum_{i'\in \Nc_i} [\W]_{ii'} \x_{j,k}}} - \gamma \y_{i,k}$;
		\IF {mod(k,q)=0}
		\STATE Calculate $\g_{i,k+1} =  \nabla f_i(\x_{i,k})$;
		\ELSE
		\STATE Calculate 
		$\g_{i,k+1} = \g_{i,k} + \frac{1}{|\mathcal{S}_{i,k}|}\sum_{j \in\mathcal{S}_{i,k}} (\nabla f_{i}(\x_{i,k+1}; \zeta^{(j)})- \nabla f_{i}(\x_{i,k}; \zeta^{(j)}))$;
		\ENDIF
		\STATE Gradient Tracking $\y_{i,k+1} = {{\sum_{j\in \Nc_{i'}} [\W]_{ii'} \y_{j,k}}}  +\g_{i,k+1} - \g_{i,k}$;
		\ENDFOR
		\ENDFOR
	\end{algorithmic}
\end{algorithm}

\begin{thm}[Convergence of Defra]\label{Thm}
	Under Assumption~\ref{Assumption: function}, with a constant-level step-size $\gamma$, which satisfies 
\\$\gamma \leq \min\{ \frac{1}{3L}, \sqrt{\frac{1-\lambda}{72m L^2}},\sqrt{\frac{1}{24m L^2}},  \frac{1}{5},  \frac{1}{40L^2},\frac{1-\lambda}{120L^2}, \frac{(1-\lambda)^2}{3}, \frac{1-\lambda}{6 L}, \sqrt{ \frac{1-\lambda}{12 L^2}}  \} $.  in Algorithm.\ref{Algorithm}( See detailed requirements in our Appendix), 
	then Algorithm~\ref{Algorithm} has the following convergence result:
	\begin{align*}
		\frac{1}{T} \sum_{t= 0}^{T-\!1}\Eb\big[{\Big\|\nabla f(\xb_k) \Big\|^2} + {\frac{1}{m}\sum_{i=1}^{m}\|\x_{i,k}- \xb_k\|^2} \big] =\mathcal{O}(\frac{1}{T})
	\end{align*}
\end{thm}

\section{Case Study}\label{sec: case}

In this section, the proposed method is compared with other baseline service representation methods, graph learning methods and training methods. All the models were trained using CPU AMD 5950X and GPU NVDIA RTX 3090.

\subsection{Dataset}
In this study, a dataset of 238,846 claims of 1,005 patients was used to validate the proposed approach. In the graph model, there are 1,005 patient nodes, 2,233 service nodes, and 15,044 doctor nodes. The dataset was collected during 2010 to 2015. Each patient has a record range from 750 days to 1,750 days. The patient record length is distributed between 100 words and 1,500 words.

\subsection{Comparisons with Other Service Representation Methods}
To investigate which service embedding method is more suitable for this application, skip-gram, GloVe and BERT are implemented. Then the embedding are used as the input of a LSTM model for a binary classification of Alzheimer disease. The GloVe model is loaded with the pretrained parameters on Common Crawl which has 840B tokens and 2.2M vocab \cite{pennington2014glove}. Each word is converted to a 300-dimension vector. Then the model is fine-tuned with LSTM for the classification. To boost the performance of BERT, Med-BERT \cite{rasmy2021med}, a contextualized embedding model pretrained on a structured EHR dataset of 28,490,650 patients, is used to obtain the service embedding. Then the BERT is fine-tuned with LSTM. For the train-test split, 65\%, 15\% and 20\% of the 1005 patients are randomly selected as train, validation and test set. 

Table \ref{tab: prelim} showed that BERT-LSTM substantially improves the prediction accuracy by up to 9\%, boosting the AUC by up to 6\%. It should be noted that only the service sequences of patients are used for this preliminary experiment. It proves the effectiveness of the context-based embedding method in the clinic application. In the following experiments, the graph information including patient attribute, doctor attribute and edge information is used. And BERT is used to obtain the service embedding.

\begin{table}[h]
\begin{center}
\caption{Comparisons among different service embedding methods for the binary classification of Alzheimer disease} \label{tab: prelim}%
\begin{tabular}{@{}rll@{}}
\toprule
Methods  & Test recall & Test AUC\\
\midrule
SKip-gram-LSTM	& 0.53	& 0.54 \\
GLoVE-LSTM	    & 0.57	& 0.55 \\
\textbf{BERT-LSTM}	    & \textbf{0.63}	& \textbf{0.61} \\
\hline
\end{tabular}
\end{center}
\end{table}

\subsection{Comparisons with Other Graph Learning Methods} \label{sec: gl}
To validate the effectiveness of our proposed model, we compare HGAT with three baseline models: graph convolutional network (GCN), graph neural network (GNN) and graph attention network (GAT). All the four models used 64 hidden neuron, 128 embedding dimension and 4 layers of the modules. In GAT and HGAT, the number of head of the multi-head attention is set as 4. For the train-test split, 65\%, 15\% and 20\% of the 1,005 patients are randomly selected as train, validation and test set. In the training of the model, only EHR in train set are used to generate the graph. In the validation and testing, the validation set and test set are incorporated into the train graph and generate the prediction. For each patient, the first 65\% record are used to construct the graph and the remaining 35\% is masked. In the testing, the model is aimed to predict the 5$\sim$10 ground-truth doctors from a set of 200$\sim$350 doctors.

The results are shown in Table \ref{tab: HGAT}. In terms of recall, HGAT can obtain 70\%, meaning that the model can find 70\% of the doctors appeared in the future record. We can see that although the HGAT is only 2\% higher than GAT for test recall, the AUC of HGAT is 6\% higher than that of GAT. And the performance of attention-based graph learning models, i.e., GAT and HGAT, is superior to the others.

\begin{table}[h]
\begin{center}
\caption{Comparisons among GCN, GNN, GAT and HGAT} \label{tab: HGAT}%
\begin{tabular}{@{}lll@{}}
\toprule
Methods  & Test recall & Test AUC\\
\midrule

GNN \cite{scarselli2008graph} & 0.63	& 0.65 \\
GCN \cite{kipf2016semi}	& 0.66	& 0.64 \\
GAT \cite{velickovic2017graph} & 0.68	& 0.70 \\
\textbf{HGAT}& \textbf{0.70}	& \textbf{0.76} \\
\hline
\end{tabular}
\end{center}
\end{table}

\subsection{Comparisons with Global Training and Local Training}

To simulate the situation that the EHR data is stored in different centers, we divided all the patients into six groups based on their location (state level). Each group have about 145, 158, 177, 207, 147 and 171 patients,respectively. It is assumed that the EHR data in each group is under an agreement and stored in one data center. A consensus matrix $\mathbf{W}$ is generated to represent the connection between the six groups as in Eq. \ref{Eq: cm}. The train-test split and labeling are the same as described in Sec. \ref{sec: gl}. The base model is HGAT used in \ref{sec: gl}.

\begin{equation}\label{Eq: cm}
\mathbf{W}=
\begin{bmatrix}
0.64 & 0.18 & 0 & 0  & 0.18 & 0\\
0.18 & 0.64 & 0 & 0  & 0 & 0.18\\
0 & 0 & 0.82 & 0.18  & 0 & 0\\
0 & 0 & 0.18 & 0.64  & 0  & 0.18\\
0.18 & 0 & 0 & 0  & 0.82 & 0 \\
0 & 0.18 & 0 & 0.18  &  0 & 0.64\\
\end{bmatrix}
\end{equation}

The comparison results are shown in Table \ref{tab: FDL}. Compared to the local training, the performance of FDL is superior. For example,  for region 1, 5 and 6, the test recall and test AUC of FDL is 14\%, 19\% higher than that of the local training, respectively. The reason is that in the local training, the generalization ability of the model is limited due to the lack of train data. Compared to the global training, FDL can achieve a comparable performance. The difference in test recall and test AUC is between 2\% and 5\%. Our proposed FDL yields the fictitious fusion center’s performance. The proof of FDL convergence is given in the supplementary section.

\begin{table}[h]
\begin{center}
\caption{Performance of federated decentralized learning on the six regions. Local: each local model is trained using the local graph only; Global: a fusion center is trained using the global graph. FDL: each local model is trained using proposed federated decentralized learning method, namely, Algorithm \ref{Algorithm}. The Recall and AUC are the model performance on the test dataset. } \label{tab: FDL}%
\begin{tabular}{@{}rllrll@{}}
\toprule
Methods  & Recall & AUC & Methods  & Recall & AUC\\
\midrule
Local-region1	& 0.66	& 0.54 & Local-region4	& 0.66	& 0.74 \\
Global-region1  & 0.73	& 0.74 & Global-region4	& 0.76	& 0.71 \\
FDL-region1     & 0.70	& 0.73 & FDL-region4	& 0.72	& 0.64 \\
\hline
Local-region2 & 0.72	& 0.78   & Local-region5	& 0.51	& 0.53 \\
Global-region2 & 0.70 & 0.77    & Global-region5	& 0.65	& 0.66 \\
FDL-region2   & 0.71	& 0.77  & FDL-region5	& 0.65	& 0.60 \\
\hline
Local-region3& 0.65	& 0.76     & Local-region6	& 0.62	& 0.58 \\
Global-region3& 0.68 & 0.79    & Global-region6	& 0.68	& 0.80 \\
FDL-region3& 0.72	& 0.76     & FDL-region6	& 0.66	& 0.75 \\
\hline
\end{tabular}
\end{center}
\end{table}

\section{Conclusion} \label{sec: conclusion}

Accurate doctor recommendation is the key to developing telehealth services and improving running efficiency of the healthcare system. This paper proposed a heterogeneous graph based model for doctor recommendation using EHR.  First, a latent vector representing the edge between two types of heterogeneous nodes is incorporated into the message passing. Second, to enrich the representation of the doctor, patient and service node, three relationships, i.e., service-service, doctor-speciality, and patient-doctor, are reconstructed simultaneously. Moreover, to solve the data privacy issue in the medical area, a federated decentralized learning is developed to help improve the local model performance. The results show that the proposed model FD-GATDR can achieve a high prediction accuracy. In future, we will use a bi-level optimization model to improve the model prediction accuracy as well as the communication efficiency for more complex heterogeneous graph.

\bibliographystyle{ACM-Reference-Format}
\input{fd_gatdr_workshop.bbl}


\onecolumn
\newpage
\appendix
\section{Supporting Lemma}
\begin{lem}[Iterates Contraction]
The following contraction properties of the iterates hold:
\begin{align} \label{iteratescontraction}
\|\x_k-\ot\bar{\x}_k\|^2 \le & (1+c_1)\lambda^2\|\x_{k-1} -\ot\bar{\x}_{k-1} \|^2 + (1+\frac{1}{c_1}) \gamma^2\|\y_{k-1}-\ot \bar\y_{k-1}\|^2, \\
\|\y_k-\ot\bar\y_k\|^2 \le& 
(1+ c_2)\lambda^2\|\y_{k-1} - \ot \bar\y_{k-1}\|^2 + (1 + \frac{1}{ c_2})\|\g_k-\g_{k-1} \|^2,
\end{align}
where $c_1$ and $c_2$ are arbitrary positive constants.
Additionally, we have 
\begin{align}\label{eqs30}
\|\x_k-\x_{k-1}\|^2 
& \le
8\|(\x_{k-1} - \ot\bar{\x}_{k-1}) \|^2 + 4\gamma^2 \|\y_{k-1} - \ot\bar\y_{k-1}\|^2 + 4\gamma^2m\|\bar\y_{k-1}\|^2.
\end{align}
\end{lem}

\begin{proof}
 Define $\Wt = \W \otimes \I_m$. First for the iterates $\x_k$, we have the following contraction:
\begin{align}\label{eqs32}
\|\Wt\x_{k} -\ot\bar\x_{k} \|^2 = \|\Wt(\x_{k} -\ot\bar\x_{k}) \|^2 \le \lambda^2\|\x_{k} -\ot\bar{\x}_{k}\|^2.
\end{align}
This is because $\x_{k} -\ot\x_{k}$ is orthogonal $\1,$ which is the eigenvector corresponding to the largest eigenvalue of $\Wt,$ and $\lambda = \max\{|\lambda_2|,|\lambda_m|\}.$
Recall that $\bar{\x}_k = \bar{\x}_{k-1} - \gamma\bar\y_{k-1},$ hence,
\begin{align}
&
\|\x_k-\ot\bar{\x}_k\|^2
=
\|\Wt\x_{k-1} - \gamma\bar\y_{k-1}-\ot(\bar{\x}_{k-1} - \gamma\bar\y_{k-1})\|^2 \notag\\
&
\stackrel{(a)}{\leq}
(1+c_1)\|\Wt\x_{k-1} -\ot\bar{\x}_{k-1} \|^2 + (1+\frac{1}{c_1}) \gamma^2\|\bar\y_{k-1}-\ot \bar\y_{k-1}\|^2 \notag\\
&
\stackrel{(b)}{\leq}
(1+c_1)\lambda^2\|\x_{k-1} -\ot\bar{\x}_{k-1} \|^2 + (1+\frac{1}{c_1}) \gamma^2\|\bar\y_{k-1}-\ot \bar\y_{k-1}\|^2, 
\end{align}
where (a) is becase of triangle inequality and (b) is from eqs.(\ref{eqs32}).

For $\bar\y_k$, we have
\begin{align}
&\|\y_k-\ot\bar\y_k\|^2 
\notag\\=&
\|\Wt \y_{k-1} + \g_k-\g_{k-1}  - \ot \big(\bar\y_{k-1} +\bar\g_k-\bar\g_{k-1} \big)\|^2 \notag\\
\le&
(1+ c_2)\lambda^2\|\y_{k-1} - \ot \bar\y_{k-1}\|^2+ (1 + \frac{1}{ c_2})\|\g_k-\g_{k-1}-\ot\big(\bar\g_k-\bar\g_{k-1} \big)\|^2\notag\\
\le&
(1+ c_2)\lambda^2\|\y_{k-1} - \ot \bar\y_{k-1}\|^2 + (1 + \frac{1}{ c_2})\|\big(\I - \frac{1}{n}(\1\1^\top)\otimes\I\big)\big(\g_k-\g_{k-1}\big)\|^2\notag\\
\stackrel{(a)}{\le}&
(1+ c_2)\lambda^2\|\y_{k-1} - \ot \bar\y_{k-1}\|^2 + (1 + \frac{1}{ c_2})\|\g_k-\g_{k-1} \|^2,
\end{align}
where (a) is due to $\|\I-\frac{1}{m}(\1\1^\top)\otimes \I \|\le 1.$ 

According to the updating
\begin{align}\label{16}
&
\|\x_k-\x_{k-1}\|^2 
= 
\|\Wt \x_{k-1} -\gamma\y_{k-1} - \x_{k-1}\|^2 \notag\\
=
&
\|(\Wt -\I)\x_{k-1} - \gamma\y_{k-1}\|^2
\le 2\|(\Wt -\I)\x_{k-1} \|^2 + 2\gamma^2 \|\y_{k-1}\|^2 \notag\\
=
&
2\|(\Wt -\I)(\x_{k-1} - \ot\bar{\x}_{k-1}) \|^2 + 2\gamma^2 \|\y_{k-1}\|^2 \notag\\
\stackrel{}{\le}
&
8\|(\x_{k-1} - \ot\bar{\x}_{k-1}) \|^2 + 4\gamma^2 \|\y_{k-1} - \ot\bar\y_{k-1}\|^2 + 4\gamma^2m\|\bar\y_{k-1}\|^2.
\end{align}
    
\end{proof}

\section{Proof of Theorem}

Next, we provide the proofs of Theorem \ref{Thm}.

\begin{proof}
	According to the algorithm update, we have:
	\begin{align}\label{eqs1}
		&
		f(\bar{\x}_{k+1}) - f(\bar{\x}_{k}) 
		\stackrel{(a)}{\le}
		\langle \nabla f(\bar{\x}_k), \bar{\x}_{k+1} - \bar{\x}_{k} \rangle + \frac{L}{2}\|\bar{\x}_{k+1} - \bar{\x}_{k}\|^2 \notag\\
		\stackrel{(b)}{=} 
		&
		-\gamma \langle \nabla f(\bar{\x}_k), \bar{\y}_{k} \rangle + \frac{L\gamma^2}{2}\|\bar{\y}_{k}\|^2 \notag\\
		= 
		&
		-\frac{\gamma}{2} \|\nabla f(\bar{\x}_k)\|^2 - (\frac{\gamma}{2} - \frac{L\gamma^2}{2})\|\bar{\y}_{k}\|^2  + \frac{\gamma}{2} \|\nabla f(\bar{\x}_k) - \bar{\y}_{k}\|^2 \notag\\
		=&-\frac{\gamma}{2} \|\nabla f(\bar{\x}_k)\|^2 - (\frac{\gamma}{2} - \frac{L\gamma^2}{2})\|\bar{\y}_{k}\|^2  + \frac{\gamma}{2} \|\nabla f(\bar{\x}_k) -\frac{1}{m} \sum_{i=1}^m \nabla f(\x_{i,k})+\frac{1}{m} \sum_{i=1}^m \nabla f(\x_{i,k})  - \bar{\y}_{k}\|^2 \notag\\
	\stackrel{(c)}{\le}&-\frac{\gamma}{2} \|\nabla f(\bar{\x}_k)\|^2 - (\frac{\gamma}{2} - \frac{L\gamma^2}{2})\|\bar{\y}_{k}\|^2  + 
		{\gamma}\|\nabla f(\bar{\x}_k) -\frac{1}{m} \sum_{i=1}^m\nabla f(\x_{i,k})\|^2+{\gamma}\|\frac{1}{m} \sum_{i=1}^m \nabla f(\x_{i,k})  - \bar{\y}_{k}\|^2 \notag\\
		\stackrel{(d)}{\le}&-\frac{\gamma}{2} \|\nabla f(\bar{\x}_k)\|^2 - (\frac{\gamma}{2} - \frac{L\gamma^2}{2})\|\bar{\y}_{k}\|^2  + 
		\frac{{\gamma}}{m} \sum_{i=1}^m\|\nabla f(\bar{\x}_k) - \nabla f(\x_{i,k})\|^2+{\gamma}\|\frac{1}{m} \sum_{i=1}^m \nabla f(\x_{i,k})  - \bar{\y}_{k}\|^2 \notag\\
		\stackrel{(e)}{\le}&-\frac{\gamma}{2} \|\nabla f(\bar{\x}_k)\|^2 - (\frac{\gamma}{2} - \frac{L\gamma^2}{2})\|\bar{\y}_{k}\|^2  + 
		\frac{{\gamma L^2}}{m} \sum_{i=1}^m\|{\x}_{i,k} - \bar\x_k\|^2+{\gamma}\|\frac{1}{m} \sum_{i=1}^m \nabla f(\x_{i,k})  - \bar{\g}_{k}\|^2 \notag\\
				\stackrel{(f)}{\le}&-\frac{\gamma}{2} \|\nabla f(\bar{\x}_k)\|^2 - (\frac{\gamma}{2} - \frac{L\gamma^2}{2})\|\bar{\y}_{k}\|^2  + 
		\frac{{\gamma L^2}}{m} \sum_{i=1}^m\|{\x}_{i,k} - \bar\x_k\|^2+{\gamma}\frac{1}{m} \sum_{i=1}^m \|\nabla f(\x_{i,k})  - \g_{i,k}\|^2 \notag\\
	\end{align}
	where (a) is because of Lipschitz continuous gradients of $f$, (b) follows from the update rule of $\x$. (c), (d) and (f) are follow from the triangle inequality. (e) is because $\bar\y_t=\bar\g_t$ and Lipschitz continuous gradients of $f$

Next, we bound the error of the gradient estimators as the follows:
\begin{align}\label{eqs18}
\mathbb{E}&\left\|\g_{i,t}-\nabla f\left(\x_{i,t}\right)\right\|^{2} 	\stackrel{(a)}{\le} \frac{L^{2}}{|\mathcal{S}_{i,k}|} \mathbb{E}\left\|\x_{i,k}-\x_{i,k-1}\right\|^{2}+\mathbb{E}\left\|\g_{i,k-1}-\nabla f\left(\x_{i,k-1}\right)\right\|^{2}\notag\\
	\stackrel{(b)}{\le}&\sum_{t=\left(n_{k}-1\right) q}^{k-1} \frac{L^{2}}{|\mathcal{S}_{i,t}|} \mathbb{E}\left\|\x_{i,t+1}-\x_{i,t}\right\|^{2}+\mathbb{E}\left\|\g_{i,\left(n_{k}-1\right) q}-\nabla f\left(\x_{i,\left(n_{k}-1\right) q}\right)\right\|^{2} \notag\\  \stackrel{(c)}{\le}
& \sum_{t=\left(n_{k}-1\right) q}^{k} \frac{L^{2}}{|\mathcal{S}_{i,t}|}\ \mathbb{E}\left\|\x_{i,t+1}-\x_{i,t}\right\|^{2}+\mathbb{E}\left\|\g_{i,\left(n_{k}-1\right) q}-\nabla f\left(\x_{i,\left(n_{k}-1\right) q}\right)\right\|^{2} \notag\\ \stackrel{(d)}{=}
& \sum_{t=\left(n_{k}-1\right) q}^{k} \frac{L^{2}}{|\mathcal{S}_{i,t}|}\ \mathbb{E}\left\|\x_{i,t+1}-\x_{i,t}\right\|^{2},
\end{align}
where (a) follows from Lemma 1 in \cite{fang2018spider}, (b)is telescoping the result of (a) over $k$ from $\left(n_{k}-1\right) q+1$ to $k$, where $k \leq n_{k} q-1$. (c) extends $k-1$ to $k$ and (d) is because our algorithm calculate exact full gradients every $q$ iterations.

Next, taking the expectation of both sides of eqs.(\ref{eqs1}) and plugging the result in eqs.(\ref{eqs18}), we have
	\begin{align}\label{eqs19}
	\mathbb{E}	f(\bar{\x}_{k+1}) - 	&	\mathbb{E}f(\bar{\x}_{k}) 
\le -\frac{\gamma}{2} 	\mathbb{E}\|\nabla f(\bar{\x}_k)\|^2 - (\frac{\gamma}{2} - \frac{L\gamma^2}{2})	\mathbb{E}\|\bar{\y}_{k}\|^2 \notag\\& + 
		\frac{{\gamma L^2}}{m} \sum_{i=1}^m	\mathbb{E}\|{\x}_{i,k} - \bar\x_k\|^2+{\gamma}\frac{1}{m} \sum_{i=1}^m 	\mathbb{E}\|\nabla f(\x_{i,k})  - \g_{i,k}\|^2\notag\\
\le& -\frac{\gamma}{2} 	\mathbb{E}\|\nabla f(\bar{\x}_k)\|^2 - (\frac{\gamma}{2} - \frac{L\gamma^2}{2})	\mathbb{E}\|\bar{\y}_{k}\|^2 \notag\\& + 
		\frac{{\gamma L^2}}{m} \sum_{i=1}^m	\mathbb{E}\|{\x}_{i,k} - \bar\x_k\|^2+{\gamma}\frac{1}{m} \sum_{i=1}^m 	\sum_{t=\left(n_{k}-1\right) q}^{k} \frac{L^{2}}{|\mathcal{S}_{i,t}|}\ \mathbb{E}\left\|\x_{i,t+1}-\x_{i,t}\right\|^{2}
	\end{align}

Next, telescoping eqs.\ref{eqs19}  over $k$ from $\left(n_{k}-1\right) q$ to $k$ where $k \leq n_{k} q-1$ and since $q=|\mathcal{S}_{i,t}|=\lceil \sqrt{n}\rceil$, we have
	\begin{align}\label{eqs20}
&	\mathbb{E}	f(\bar{\x}_{k+1}) - 		\mathbb{E}f(\bar{\x}_{(n_k-1)q} ) 
\le -\frac{\gamma}{2} 	\mathbb{E}\sum_{t=\left(n_{k}-1\right) q}^k \|\nabla f(\bar{\x}_{t})\|^2 - (\frac{\gamma}{2} - \frac{L\gamma^2}{2})	\sum_{t=\left(n_{k}-1\right) q}^k \mathbb{E}\|\bar{\y}_{t}\|^2 \notag\\& + 
		\frac{{\gamma L^2}}{m} \sum_{t=\left(n_{k}-1\right) q}^k\sum_{i=1}^m	\mathbb{E}\|{\x}_{i,t} - \bar\x_t\|^2+{\gamma}\frac{1}{m} \sum_{i=1}^m \sum_{j=\left(n_{k}-1\right) q}^k	\sum_{t=\left(n_{k}-1\right) q}^{j} \frac{L^{2}}{|\mathcal{S}_{i,t}|}\ \mathbb{E}\left\|\x_{i,t+1}-\x_{i,t}\right\|^{2}\notag\\
&\le -\frac{\gamma}{2} 	\mathbb{E}\sum_{t=\left(n_{k}-1\right) q}^k \|\nabla f(\bar{\x}_{t})\|^2 - (\frac{\gamma}{2} - \frac{L\gamma^2}{2})	\sum_{t=\left(n_{k}-1\right) q}^k \mathbb{E}\|\bar{\y}_{t}\|^2 \notag\\& + 
		\frac{{\gamma L^2}}{m} \sum_{t=\left(n_{k}-1\right) q}^k\sum_{i=1}^m	\mathbb{E}\|{\x}_{i,t} - \bar\x_t\|^2+{\gamma}\frac{1}{m} \sum_{i=1}^m \sum_{j=\left(n_{k}-1\right) q}^k	\sum_{t=\left(n_{k}-1\right) q}^{k} \frac{L^{2}}{|\mathcal{S}_{i,t}|}\ \mathbb{E}\left\|\x_{i,t+1}-\x_{i,t}\right\|^{2}\notag\\
	&= -\frac{\gamma}{2} 	\mathbb{E}\sum_{t=\left(n_{k}-1\right) q}^k \|\nabla f(\bar{\x}_{t})\|^2 - (\frac{\gamma}{2} - \frac{L\gamma^2}{2})	\sum_{t=\left(n_{k}-1\right) q}^k \mathbb{E}\|\bar{\y}_{t}\|^2 \notag\\& + 
		\frac{{\gamma L^2}}{m} \sum_{t=\left(n_{k}-1\right) q}^k\sum_{i=1}^m	\mathbb{E}\|{\x}_{i,t} - \bar\x_t\|^2+{\gamma}\frac{1}{m} \sum_{i=1}^m \sum_{t=\left(n_{k}-1\right) q}^k	q \frac{L^{2}}{|\mathcal{S}_{i,t}|}\ \mathbb{E}\left\|\x_{i,t+1}-\x_{i,t}\right\|^{2}\notag\\
	&= -\frac{\gamma}{2} 	\mathbb{E}\sum_{t=\left(n_{k}-1\right) q}^k \|\nabla f(\bar{\x}_{t})\|^2 - (\frac{\gamma}{2} - \frac{L\gamma^2}{2})	\sum_{t=\left(n_{k}-1\right) q}^k \mathbb{E}\|\bar{\y}_{t}\|^2 \notag\\& + 
		\frac{{\gamma L^2}}{m} \sum_{t=\left(n_{k}-1\right) q}^k\sum_{i=1}^m	\mathbb{E}\|{\x}_{i,t} - \bar\x_t\|^2+{\gamma} \sum_{t=\left(n_{k}-1\right) q}^k L^2  \mathbb{E}\left\|\x_{t+1}-\x_{t}\right\|^{2}.
	\end{align}

Next, combing the result of eqs.(\ref{eqs20}), eqs.(\ref{iteratescontraction})- (\ref{eqs30}) and telescoping over $k$ from $0$ to $K$, we have
	\begin{align}\label{eqs21}\notag\\
&	\mathbb{E}	f(\bar{\x}_{K+1}) - \mathbb{E}f(\bar{\x}_{0}) 
	+\|\x_K-\ot\bar{\x}_K\|^2-\|\x_0-\ot\bar{\x}_0\|^2 \notag\\&+\gamma[ \|\y_K-\ot\bar{\y}_K\|^2-\|\y_0-\ot\bar{\y}_0\|^2 ] \notag\\
\le& -\frac{\gamma}{2} \sum_{k=0}^K	\mathbb{E} \|\nabla f(\bar{\x}_k)\|^2 - (\frac{\gamma}{2} - \frac{L\gamma^2}{2})	\sum_{k=0}^K	\mathbb{E}\|\bar{\y}_{k}\|^2 \notag\\& + 
	{	\gamma L^2} \sum_{k=0}^K		\mathbb{E}\|{\x}_{k} - \bar\x_k\|^2+{\gamma} \sum_{k=0}^KL^2  \mathbb{E}\left\|\x_{k+1}-\x_{k}\right\|^{2}\notag\\
		&+((1+c_1)\lambda^2 -1)\sum_{k=0}^K\mathbb{E} \|\x_{k} -\ot\bar{\x}_{k} \|^2 + (1+\frac{1}{c_1}) \gamma^2 \sum_{k=0}^K\mathbb{E}\|\y_{k}-\ot \bar\y_{k}\|^2\notag\\
		&+((1+ c_2)\lambda^2-1)\gamma \sum_{k=0}^K \mathbb{E}\|\y_{k} - \ot \bar\y_{k}\|^2 + (1 + \frac{1}{ c_2}) \gamma \sum_{k=0}^K \mathbb{E}\|\g_{k+1}-\g_{k} \|^2.
	\end{align}

For $\mathbb{E}\|\g_k-\g_{k-1} \|^2 $ and let $\nabla f_t = [\nabla f({{\x}}_{1,t})^\top,\cdots,\nabla f({{\x}}_{m,t})^\top]^{\top} $, we have
\begin{align}\label{22}
&
\Eb\|\g_k-\g_{k-1} \|^2 = \Eb\|\g_k -\nabla_{{{\x}}}f_t + \nabla_{{{\x}}}f_t - \nabla_{{{\x}}}f_{t-1} + \nabla_{{{\x}}}f_{t-1} - \v_{t-1}\|^2  \notag\\
\le
& 3\Eb\|\g_k - \nabla_{{{\x}}}f_{t}\|^2  + 3 \Eb\|\nabla_{{{\x}}} f_{t}  - \nabla_{{{\x}}} f_{t-1}\|^2 + 3\Eb\|\nabla_{{{\x}}} f_{t-1} - \g_{k-1}\|^2 \notag\\
\le
& 0+ 3L^2 \Eb\|{{\x}}_{t}-{{\x}}_{t-1}\|^2 +0 .
\end{align}

Next, combing the result of eqs.(\ref{eqs21}), eqs.(\ref{22}) and eqs.(\ref{16}), we have
	\begin{align}\label{eqs23}\notag\\
&	\mathbb{E}	f(\bar{\x}_{K+1}) - \mathbb{E}f(\bar{\x}_{0}) 
	+\|\x_K-\ot\bar{\x}_K\|^2-\|\x_0-\ot\bar{\x}_0\|^2 \notag\\&+\gamma[ \|\y_K-\ot\bar{\y}_K\|^2-\|\y_0-\ot\bar{\y}_0\|^2 ] \notag\\
\le& -\frac{\gamma}{2} \sum_{k=0}^K	\mathbb{E} \|\nabla f(\bar{\x}_k)\|^2 - (\frac{\gamma}{2} - \frac{L\gamma^2}{2})	\sum_{k=0}^K	\mathbb{E}\|\bar{\y}_{k}\|^2 \notag\\& + 
	{	\gamma L^2} \sum_{k=0}^K		\mathbb{E}\|{\x}_{k} - \bar\x_k\|^2+{\gamma} \sum_{k=0}^KL^2  \mathbb{E}\left\|\x_{k+1}-\x_{k}\right\|^{2}\notag\\
		&+((1+c_1)\lambda^2 -1)\sum_{k=0}^K\mathbb{E} \|\x_{k} -\ot\bar{\x}_{k} \|^2 + (1+\frac{1}{c_1}) \gamma^2 \sum_{k=0}^K\mathbb{E}\|\y_{k}-\ot \bar\y_{k}\|^2\notag\\
		&+((1+ c_2)\lambda^2-1)\gamma \sum_{k=0}^K \mathbb{E}\|\y_{k} - \ot \bar\y_{k}\|^2 + (1 + \frac{1}{ c_2}) \gamma 3 L^2 \sum_{k=0}^K \mathbb{E}\|\x_{k+1}-\x_{k} \|^2\notag\\
\le& -\frac{\gamma}{2} \sum_{k=0}^K	\mathbb{E} \|\nabla f(\bar{\x}_k)\|^2 - (\frac{\gamma}{2} - \frac{L\gamma^2}{2})	\sum_{k=0}^K	\mathbb{E}\|\bar{\y}_{k}\|^2 \notag\\& + 
	{	\gamma L^2} \sum_{k=0}^K		\mathbb{E}\|{\x}_{k} - \bar\x_k\|^2
		+((1+c_1)\lambda^2 -1)\sum_{k=0}^K\mathbb{E} \|\x_{k} -\ot\bar{\x}_{k} \|^2\notag\\ &+ (1+\frac{1}{c_1}) \gamma^2 \sum_{k=0}^K\mathbb{E}\|\y_{k}-\ot \bar\y_{k}\|^2
		+((1+ c_2)\lambda^2-1)\gamma \sum_{k=0}^K \mathbb{E}\|\y_{k} - \ot \bar\y_{k}\|^2 \notag\\
		&\!+\! [(1 \!+\! \frac{1}{ c_2}) \gamma 3 L^2 \!+ \!\gamma L^2] \sum_{k=0}^K \mathbb{E}[8\|(\x_{k} \!-\!\! \ot\!\bar{\x}_{k}) \|^2 \!+\! 4\gamma^2 \|\y_{k}\! \!-\! \!\ot\!\bar\y_{k}\|^2 \!+\! 4\gamma^2m\|\bar\y_{k}\|^2]\notag\\
= & -\frac{\gamma}{2} \sum_{k=0}^K	\mathbb{E} \|\nabla f(\bar{\x}_k)\|^2 - (\frac{\gamma}{2} - \frac{L\gamma^2}{2}-4\gamma^2m [(1 \!+\! \frac{1}{ c_2}) \gamma 3 L^2 \!+ \!\gamma L^2] )	\sum_{k=0}^K	\mathbb{E}\|\bar{\y}_{k}\|^2 \notag\\& + 
		\{ (1+c_1)\lambda^2 -1+	{	\gamma L^2} +8 [(1 \!+\! \frac{1}{ c_2}) \gamma 3 L^2 \!+ \!\gamma L^2] \}\sum_{k=0}^K\mathbb{E} \|\x_{k} -\ot\bar{\x}_{k} \|^2\notag\\ &+ 
		[((1+ c_2)\lambda^2-1)\gamma+ (1+\frac{1}{c_1}) \gamma^2 + 4\gamma^2[(1 \!+\! \frac{1}{ c_2}) \gamma 3 L^2 \!+ \!\gamma L^2] ] \sum_{k=0}^K \mathbb{E}\|\y_{k} - \ot \bar\y_{k}\|^2\notag\\
= & -\frac{\gamma}{2} \sum_{k=0}^K	\mathbb{E} \|\nabla f(\bar{\x}_k)\|^2\! +\! C_1	\sum_{k=0}^K	\mathbb{E}\|\bar{\y}_{k}\|^2 \!+\!C_2 \sum_{k=0}^K\mathbb{E} \|\x_{k}\! -\!\ot\!\bar{\x}_{k} \|^2\!+\!C_3 \sum_{k=0}^K \mathbb{E}\|\y_{k} \!-\! \ot\! \bar\y_{k}\|^2,
	\end{align}
where 
	\begin{align}
C_1=& -(\frac{\gamma}{2} - \frac{L\gamma^2}{2}-4\gamma^2m [(1 \!+\! \frac{1}{ c_2}) \gamma 3 L^2 \!+ \!\gamma L^2] ) \notag\\
C_2=& \{ (1+c_1)\lambda^2 -1+	{	\gamma L^2} +8 [(1 \!+\! \frac{1}{ c_2}) \gamma 3 L^2 \!+ \!\gamma L^2] \} \notag\\
C_3=& 	[((1+ c_2)\lambda^2-1)\gamma+ (1+\frac{1}{c_1}) \gamma^2 + 4\gamma^2[(1 \!+\! \frac{1}{ c_2}) \gamma 3 L^2 \!+ \!\gamma L^2] ].
	\end{align}

Next, let $c_1=c_2=\frac{1}{\lambda}-1$, we have
	\begin{align}
C_1=& (- \frac{\gamma}{2} + \frac{L\gamma^2}{2}+4\gamma^2m [(1 \!+\! \frac{1}{ c_2}) \gamma 3 L^2 \!+ \!\gamma L^2] ) \leq - \frac{\gamma}{2}+ \frac{\gamma}{6}+ \frac{\gamma}{6}+ \frac{\gamma}{6}=0,
	\end{align}
where $\gamma \leq \min\{ \frac{1}{3L}, \sqrt{\frac{1-\lambda}{72m L^2}},\sqrt{\frac{1}{24m L^2}} \}$.
	\begin{align}
C_2=& \{ (1+c_1)\lambda^2 -1+	{	\gamma L^2} +8 [(1 \!+\! \frac{1}{ c_2}) \gamma 3 L^2 \!+ \!\gamma L^2] \} \notag\\ =& \lambda -1+	{	\gamma L^2} +8 [\frac{1}{1-\lambda} \gamma 3 L^2 \!+ \!\gamma L^2] \leq \frac{1}{5} -1 +\frac{1}{5} +\frac{1}{5} +\frac{1}{5} =-\frac{1}{5}, 
	\end{align}
where $\gamma \leq \min\{ \frac{1}{5},  \frac{1}{40L^2},\frac{1-\lambda}{120L^2} \}$.
	\begin{align}
C_3=& 	[((1+ c_2)\lambda^2-1)\gamma+ (1+\frac{1}{c_1}) \gamma^2 + 4\gamma^2[(1 \!+\! \frac{1}{ c_2}) \gamma 3 L^2 \!+ \!\gamma L^2] ]\notag\\
=& 	(\lambda-1)\gamma+ \frac{1}{1-\lambda} \gamma^2 + 4\gamma^2[\frac{1}{1-\lambda} \gamma 3 L^2 \!+ \!\gamma L^2] \notag\\\leq &-(1-\lambda)\gamma+ \frac{(1-\lambda)\gamma}{3}+ \frac{(1-\lambda)\gamma}{3}+ \frac{(1-\lambda)\gamma}{3}=0,
	\end{align}
where $\gamma \leq \min\{ \frac{(1-\lambda)^2}{3}, \frac{1-\lambda}{6 L}, \sqrt{ \frac{1-\lambda}{12 L^2}}  \}$.

Thus, we can conclude that 
	\begin{align}\label{eqs23}\notag\\
&	\mathbb{E}	f(\bar{\x}_{K+1}) - \mathbb{E}f(\bar{\x}_{0}) 
	+\|\x_K-\ot\bar{\x}_K\|^2-\|\x_0-\ot\bar{\x}_0\|^2 \notag\\&+\gamma[ \|\y_K-\ot\bar{\y}_K\|^2-\|\y_0-\ot\bar{\y}_0\|^2 ] 
\leq  -\frac{\gamma}{2} \sum_{k=0}^K	\mathbb{E} \|\nabla f(\bar{\x}_k)\|^2\! -\frac{1}{5} \sum_{k=0}^K\mathbb{E} \|\x_{k}\! -\!\ot\!\bar{\x}_{k} \|^2.
	\end{align}

Define $\mathtt{Q}_t= 	\mathbb{E}	f(\bar{\x}_{t+1})
	+\|\x_t-\ot\bar{\x}_t\|^2+\gamma[ \|\y_t-\ot\bar{\y}_t\|^2$, then we have 
	\begin{align}\label{eqs23}\notag\\
& \frac{1}{T}\sum_{t=0}^{T-1}	\mathbb{E} \|\nabla f(\bar{\x}_t)\|^2\! + \frac{1}{T}\sum_{t=0}^{T-1}\mathbb{E} \|\x_{t}\! -\!\ot\!\bar{\x}_{t} \|^2 \leq \frac{\mathtt{Q}_{T}- \mathtt{Q}_0}{T \min\{\frac{1}{5}, \frac{\gamma}{2} \}} = \mathcal{O}({\frac{1}{T}})	.
	\end{align}	

Thus, we can conclude our result in Theorem.~\ref{Thm}.
	\begin{align*}
		\frac{1}{T} \sum_{t= 0}^{T-\!1}\Eb\big[{\Big\|\nabla f(\xb_k) \Big\|^2} + {\frac{1}{m}\sum_{i=1}^{m}\|\x_{i,k}- \xb_k\|^2} \big] =\mathcal{O}(\frac{1}{T}),
	\end{align*}
where $\gamma$ is a constant-level step-size, which satisfies \\
$\gamma \leq \min\{ \frac{1}{3L}, \sqrt{\frac{1-\lambda}{72m L^2}},\sqrt{\frac{1}{24m L^2}},  \frac{1}{5},  \frac{1}{40L^2},\frac{1-\lambda}{120L^2}, \frac{(1-\lambda)^2}{3}, \frac{1-\lambda}{6 L}, \sqrt{ \frac{1-\lambda}{12 L^2}}  \} $. 
\end{proof}

\end{document}

%% file: symbols_commands.tex
\newcommand{\D}{\mathcal{D}}

\newcommand{\Eb}{\mathbb{E}}

\newcommand{\g}{\mathbf{g}}

\newcommand{\I}{\mathbf{I}}

\newcommand{\Wt}{\widetilde{\mathbf{W}}}

\newcommand{\Nc}{\mathcal{N}}

\renewcommand{\v}{\mathbf{v}}

\newcommand{\W}{\mathbf{W}}

\newcommand{\x}{\mathbf{x}}

\newcommand{\y}{\mathbf{y}}

\newcommand{\xb}{\mathbf{\bar{x}}}

\newcommand{\1}{\mathbf{1}}
\newcommand{\ot}{\1\otimes}

\newtheorem{thm}{Theorem}

\newtheorem{lem}{Lemma}

\newtheorem{assum}{Assumption}